\documentclass[12pt]{article}%
\usepackage{amsmath,amssymb}%
\usepackage{hyperref}%
\usepackage{epsfig}%
\usepackage{theorem}%
\usepackage{wide}
\newcommand{\matousek}{Matou{\v s}ek}
\newcommand{\plechac}{Plech{\' a}{\v c}}

\def\CH{{\mathop{\mathrm{CH}}}}
\def\CHX{{\mathop{\cal{CH}}}}
\newcommand{\inter}{\mathop{int}}

\newcommand{\DCH}[1]{\mathop{{#1}\text{-}{\cal CH}}}
\newcommand{\SC}{{\cal S}}
\newcommand{\HL}[1]{{{#1}_L}}
\newcommand{\HR}[1]{{{#1}_R}}

\newcommand{\D}{{\cal D}}

\newcommand{\T}{{\cal T}}
\def\CHO{{\mathop{CH}}}

\def\bd{{\partial}}
\newcommand{\MakeBig}{\rule[-.2cm]{0cm}{0.4cm}}

\newcommand{\MakeHigher}[1]{\rule[#1]{0cm}{0.00cm}}

\newcommand{\sep}[1]{\,\left|\, {#1} \MakeBig\right.}
\newcommand{\brc}[1]{\left\{ {#1} \right\}}
\newcommand{\pbrc}[1]{\left[ {#1{\MakeHigher{-0.2cm}}} \right]}
\newcommand{\pth}[1]{\left( {#1} \right)}

\newcommand{\ceil}[1]{\left\lceil {#1} \right\rceil}

\def\B{{\cal B}}
\def\C{{\cal C}}
\def\D{{\cal D}}
\def\Q{{\cal Q}}
\def\S{{\cal S}}

\newcommand{\Vol}{\mathop{\mathrm{Vol}}}
\newcommand{\sign}{\mathop{\mathrm{sign}}}
\newcommand{\QHull}[1]{\mathrm{\mathop{{#1}\text{-}{co}}}}
\newcommand{\QH}{\QHull{\Q_{sc}}}
\newcommand{\QSC}{\Q_{sc}}

\newcommand{\SQR}{{S'}}
\newcommand{\pspan}{\mathop{pspan}}
\newcommand{\ray}{\mathop{ray}}
\newcommand{\dpair}{{\D\text{-pair}}}
\newcommand{\DPAIRS}[1]{{#1}_{\text{pairs}}}
\newcommand{\QD}{{\cal{Q}}}

\newcommand{\remove}[1]{}

\newtheorem{theorem}{Theorem}[section] % section
\newtheorem{lemma}[theorem]{Lemma}

\newtheorem{definition}[theorem]{Definition}

\newtheorem{corollary}[theorem]{Corollary}

\newenvironment{proof}{{\em Proof:}}{\hfill{\hfill\rule{2mm}{2mm}}}
{\theorembodyfont{\rm} \newtheorem{remark}[theorem]{Remark}}
\renewcommand{\Re}{{\rm I\!\hspace{-0.025em} R}}

\newcommand{\atgen}{\symbol{'100}}
\newcommand{\SarielThanks}[1]{\thanks{Department of Computer
      Science; 
      University of Illinois; 
      201 N. Goodwin Avenue;
      Urbana, IL, 61801, USA;
      {\tt sariel\atgen{}uiuc.edu}; {\tt
         \url{http://www.uiuc.edu/\string~sariel/}.} #1}}

\title{On the Expected Complexity of Random Convex
   Hulls\thanks{This work has been supported by a grant from
      the U.S.--Israeli Binational Science Foundation.  This
      work is part of the author's Ph.D. thesis, prepared at
      Tel-Aviv University under the supervision of Prof.
      Micha Sharir.}}

\author{ Sariel Har-Peled\SarielThanks{} }

\date{ \today }

\begin{document}
\maketitle

\begin{abstract}
    In this paper we present several results on the expected
    complexity of a convex hull of $n$ points chosen
    uniformly and independently from a convex
    shape.     
    
    (i) We show that the expected number of vertices of the
    convex hull of $n$ points, chosen uniformly and
    independently from a disk is $O(n^{1/3})$, and $O( k
    \log{n})$ for the case a convex polygon with $k$ sides.
    Those results are well known (see
    \cite{rs-udkhv-63,r-slcdn-70,ps-cgi-85}), but we
    believe that the elementary proof given here are simpler
    and more intuitive.
    
    (ii) Let $\D$ be a set of directions in the plane, we
    define a generalized notion of convexity induced by
    $\D$, which extends both rectilinear convexity and
    standard convexity.
    
    We prove that the expected complexity of the $\D$-convex
    hull of a set of $n$ points, chosen uniformly and
    independently from a disk, is $O \pth {n^{1/3} +
       \sqrt{n\alpha(\D)}}$, where $\alpha(\D)$ is the
    largest angle between two consecutive vectors in $\D$.
    This result extends the known bounds for the cases of
    rectilinear and standard convexity.

    (iii) Let $\B$ be an axis parallel hypercube in $\Re^d$.
    We prove that the expected number of points on the
    boundary of the quadrant hull of a set $S$ of $n$
    points, chosen uniformly and independently from $\B$ is
    $O(\log^{d-1}n)$. Quadrant hull of a set of points is an
    extension of rectilinear convexity to higher dimensions.
    In particular, this number is larger than the number of
    maxima in $S$, and is also larger than the number of
    points of $S$ that are vertices of the convex hull of
    $S$. 

    Those bounds are known \cite{bkst-anmsv-78}, but we
    believe the new proof is simpler.    
\end{abstract}

\section[Introduction]{Introduction}

Let $C$ be a fixed compact convex shape, and let $X_n$ be a
random sample of $n$ points chosen uniformly and
independently from $C$. Let $Z_n$ denote the number of
vertices of the convex hull of $X_n$.  R{\'e}nyi and Sulanke
\cite{rs-udkhv-63} showed that $E[Z_n] = O(k \log{n})$, when
$C$ is a convex polygon with $k$ vertices in the plane.
Raynaud \cite{r-slcdn-70} showed that expected number of
facets of the convex hull is $O(n^{(d-1)/(d+1)})$, where $C$
is a ball in $\Re^d$, so $E[Z_n] = O(n^{1/3})$ when $C$ is a
disk in the plane.  Raynaud \cite{r-slcdn-70} showed that
the expected number of facets of $\CH(X_n) = ConvexHull(X_n)$ is $O
\pth{(\log(n))^{(d-1)/2}}$, where the points are chosen from
$\Re^d$ by a $d$-dimensional normal distribution. See
\cite{ww-sg-93} for a survey of related results.

All these bounds are essentially derived by computing or
estimating integrals that quantify the probability of two
specific points of $X_n$ to form an edge of the convex hull
(multiplying this probability by $\binom{n}{2}$ gives
$E[Z_n]$). Those integrals are fairly complicated to
analyze, and the resulting proofs are rather long,
counter-intuitive and not elementary.

Efron \cite{e-chrsp-65} showed that instead of arguing about
the expected number of vertices directly, one can argue
about the expected area/volume of the convex hull, and this
in turn implies a bound on the expected number of vertices
of the convex hull. In this paper, we present a new argument
on the expected area/volume of the convex hull (this method
can be interpreted as a discrete approximation to the
integral methods). The argument
goes as follows: Decompose $C$ the into smaller shapes
(called tiles). Using the topology of the tiling and the
underlining type of convexity, we argue about the expected
number of tiles that are exposed by the random convex hull,
where a tile is exposed if it does not lie completely in the
interior of the random convex hull.  Resulting in a lower
bound on the area/volume of the random convex hull.  We
apply this technique to the standard case, and also for
more exotic types of convexity.

In Section \ref{sec:random:ch}, we give a rather simple and
elementary proofs of the aforementioned bounds
$E[Z_n]=O(n^{1/3})$ for $C$ a disk, and $E[Z_n]=O \pth{ k
   \log{n}}$ for $C$ a convex $k$-gon.  We believe that
these new elementary proofs are indeed simpler and more
intuitive\footnote{Preparata and Shamos \cite[pp.
   152]{ps-cgi-85} comment on the older proof for the case
   of a disk: ``Because the circle has no corners, the
   expected number of hull vertices is comparatively high,
   although we know of no elementary explanation of the
   $n^{1/3}$ phenomenon in the planar case.''  It is the
   author's belief that the proof given here remedies this
   situation.} than the previous integral-based proofs.

The question on the expected complexity of the convex hull
remains valid, even if we change our type of convexity.  In
Section \ref{sec:genrelized:convex}, we define a generalized
notion of convexity induced by $\D$, a given set of
directions. This extends both rectilinear convexity, and
standard convexity.  We prove that the expected complexity
of the $\D$-convex hull of a set of $n$ points, chosen
uniformly and independently from a disk, is $O \pth {n^{1/3}
   + \sqrt{n\alpha(\D)}}$, where $\alpha(\D)$ is the largest
angle between two consecutive vectors in $\D$. This result
extends the known bounds for the cases of rectilinear and
standard convexity.

Finally, in Section \ref{sec:hcube}, we deal with another
type convexity, which is an extension of the generalized
convexity mentioned above for the higher dimensions, where
the set of the directions is the standard orthonormal basis
of $\Re^d$. We prove that the expected number of points that
lie on the boundary of the quadrant hull of $n$ points,
chosen uniformly and independently from the axis-parallel
unit hypercube in $\Re^d$, is $O(\log^{d-1} n)$. This
readily imply $O(\log^{d-1} n)$ bound on the expected number
of maxima and the expected number of vertices of the convex
hull of such a point set.  Those bounds are known
\cite{bkst-anmsv-78}, but we believe the new proof is
simpler and more intuitive.

%%%%%%%%%%%%%
%%%%%%%%%%%%% On the Complexity of the Convex Hull of a Random Point 
%%%%%%%%%%%%%    Set

\section[On the Complexity of the Convex Hull of a Random Point
Set]{On the Complexity of the Convex Hull of a Random Point
   Set}
\label{sec:random:ch}

In this section, we show that the expected number of
vertices of the convex hull of $n$ points, chosen uniformly
and independently from a disk, is $O(n^{1/3})$. Applying the
same technique to a convex polygon with $k$ sides, we prove
that the expected number of vertices of the convex hull is
$O( k \log{n})$.\footnote{As already noted, these results are well
known (\cite{rs-udkhv-63,r-slcdn-70,ps-cgi-85}), but we
believe that the elementary proofs given here are simpler
and more intuitive.} The following lemma, shows that the
larger the expected area outside the random
convex hull, the larger is the expected number of vertices
of the convex hull.

\begin{lemma}
    Let $C$ be a bounded convex set in the plane, such that
    the expected area of the convex hull of $n$ points,
    chosen uniformly and independently from $C$, is at least
    $\pth{1-f(n)}Area(C)$, where $1 \geq f(n) \geq 0$, for
    $n \geq 0$. Then the expected number of vertices of the
    convex hull is $\leq n f(n/2)$.

    \label{lemma:area:to:vertices}
\end{lemma}

\begin{proof}
    Let $N$ be a random sample of $n$ points, chosen
    uniformly and independently from $C$. Let $N_1$ (resp.
    $N_2$) denote the set of the first (resp. last) $n/2$
    points of $N$.  Let $V_1$ (resp.  $V_2$) denote the
    number of vertices of $H = \CHO(N_1 \cup N_2)$ that
    belong to $N_1$ (resp.  $N_2$), where $\CHO(N_1 \cup
    N_2) = {\mathrm{ConvexHull}}( N_1 \cup N_2 )$.
    
    Clearly, the expected number of vertices of $C$ is
    $E[V_1] + E[V_2]$. On the other hand,
    \[
    E \pbrc{V_1 \sep{N_2}} \leq \frac{n}{2} \pth{\frac{Area(C) -
          Area(\CHO(N_2))}{Area(C)}},
    \]
    since $V_1$ is bounded by the expected number of points
    of $N_1$ falling outside $\CHO(N_2)$.
    
    We have
    \begin{eqnarray*}
        E[V_1] &=& E_{N_2} \pbrc{ E[V_1|N_2] } \leq E
        \pbrc{\frac{n}{2}
           \pth{\frac{Area(C) - Area(\CHO(N_2))}{Area(C)}}}\\
        &\leq& \frac{n}{2}f(n/2),
    \end{eqnarray*}
    since $E[X]=E_Y[E[X|Y]]$ for any two random variables $X,Y$.  Thus,
    the expected number of vertices of $H$ is $E[V_1] + E[V_2] \leq
    n f(n/2)$.
\end{proof}

\begin{remark}
    Lemma \ref{lemma:area:to:vertices} is known as {\em Efron's
       Theorem}. See \cite{e-chrsp-65}.
\end{remark}

\begin{theorem}
    The expected number of vertices of the convex hull of $n$ points,
    chosen uniformly and independently from the unit disk, is
    $O(n^{1/3})$.
    
    \label{theorem:area}
\end{theorem}

\begin{proof}
    We claim that the expected area of the convex hull of
    $n$ points, chosen uniformly and independently from the
    unit disk, is at least $\pi - O \pth{ n^{-2/3}}$.
    
    Indeed, let $D$ denote the unit disk, and assume without
    loss of generality, that $n=m^3$, where $m$ is a
    positive integer.  Partition $D$ into $m$ sectors,
    $\SC_1, \ldots, \SC_m$, by placing $m$ equally spaced
    points on the boundary of $D$ and connecting them to the
    origin. Let $D_1, \ldots, D_{m^2}$ denote the $m^2$
    disks centered at the origin, such that (i) $D_{1} = D$,
    and (ii) $Area(D_{i-1}) - Area(D_{i}) = \pi/m^2$, for
    $i=2, \ldots, m^2$.  Let $r_i$ denote the radius of
    $D_i$, for $i=1, \ldots, m^2$.
    
    Let $S_{i,j} = (D_{i} \setminus D_{i+1}) \cap \SC_j$,
    and $S_{m^2,j}=D_{m^2} \cap \SC_j$, for $i=1, \ldots,
    m^2-1$, $j=1, \ldots, m$. The set $S_{i,j}$ is called
    the $i$-th {\em tile} of the sector $\SC_j$, and its
    area is $\pi/n$, for $i=1,\ldots,m^2$, $j=1, \ldots,m$.
    
    Let $N$ be a random sample of $n$ points chosen
    uniformly and independently from $D$. Let $X_j$ denote
    the first index $i$ such that $N \cap S_{i, j} \ne
    \emptyset$, for $j=1, \ldots, m$. For a fixed $j \in
    \brc{1,\ldots,m}$, the probability that
    $X_j = k$ is upper-bounded by the probability that the
    tiles $S_{1, j}
, \ldots, S_{(k-1),j}$ do not contain any
    point of $N$; namely, by $\pth{1-\frac{k-1}{n}}^{n}$.
    Thus, $P[X_j = k] \leq \pth{1-\frac{k-1}{n}}^{n} \leq
    e^{-(k-1)}$, since $1-x\leq e^{-x}$, for $x \geq 0$.
    Thus,
    \[
    E \pbrc{ X_j } = \sum_{k=1}^{m^2} k P[X_j = k ] \leq
    \sum_{k=1}^{m^2} k e^{-(k-1)} = O(1),
    \]
    for $j=1, \ldots, m$.
    
    Let $K_o$ denote the convex hull of $N \cup \brc{o}$,
    where $o$ is the origin. The tile $S_{i,j}$ is {\em
       exposed} by a set $K$, if $S_{i,j} \setminus K \ne
    \emptyset$.  We claim that at most $X_{j-1} + X_{j+1} +
    O(1)$ tiles are exposed by $K_o$ in the sector $\SC_j$,
    for $j=1,\ldots, m$ (where we put $X_0 = X_m$, $X_{m+1}
    = X_1$).
    
    \begin{figure}
        \centerline{
           \def\IPEfile{figs/expose.ipe}\input{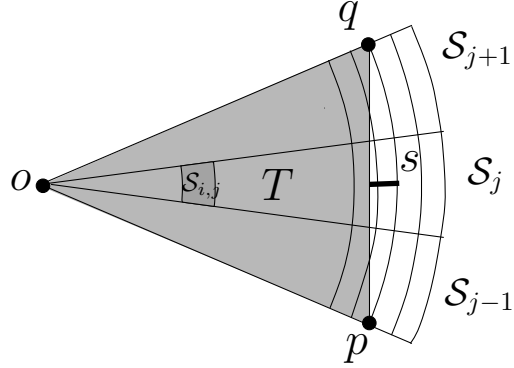}
           }
%        \comment{Add a tile called $S_{i,j}$ to the figure}

        \caption{Illustrating the proof that bounds the number of tiles
           exposed by $T$ inside $\SC_j$}
        
        \label{fig:slice}
    \end{figure}
    
    Indeed, let $w=w(N,j) = max(X_{j-1},X_{j+1})$, and let $p,q$ be
    the two points in $S_{j-1,w}, S_{j+1,w}$, respectively,
    such that the number of sets exposed by the triangle $T
    = \triangle{opq}$, in the sector $\SC_i$, is maximal.
    Both $p$ and $q$ lie on $\bd{D_{w+1}}$ and on the
    external radii bounding $\SC_{j-1}$ and $\SC_{j+1}$, as
    shown in Figure \ref{fig:slice}. Clearly, any tile which
    is exposed in $\SC_j$ by $K_o$ is also exposed by $T$.
    Let $s$ denote the segment connecting the middle of the
    base of $T$ to its closest point on $\bd{D_w}$. The
    number of tiles in $\SC_j$ exposed by $T$ is bounded by
    $\max \pth{X_{j-1}, X_{j+1}}$, plus the number of tiles
    intersecting the segment $s$.  The length of $s$ is
    \[
    |oq| - |oq| \cos \pth{ \frac{3}{2} \cdot \frac{2\pi}{m}}
    \leq 1 - \cos \pth{ \frac{3}{2} \cdot \frac{2\pi}{m}}
    \leq \frac{1}{2} \pth{ \frac{3\pi}{m}}^2 =
    \frac{4.5\pi^2}{m^2},
    \]
    since $\cos(x) \geq 1-x^2/2$, for $x \geq 0$.
    
    On the other hand, $r_{i+1} - r_{i} \geq r_i - r_{i-1}
    \geq 1/(2m^2)$, for $i=2,\ldots,m^2$. Thus, the segment
    $s$ intersects at most $\ceil{||s||/(1/(2m^2))} =
    \ceil{9\pi^2} = 89$ tiles, and we have that the number
    of tiles exposed in the sector $\SC_i$ by $K_o$ is at most
    $\max \pth{X_{j-1}, X_{j+1}} + 89 \leq X_{j-1} + X_{j+1}
    + 89$, for $j=1, \ldots, m$.

    Thus, the expected number of tiles exposed by $K_o$ is at most 
    \[
    E \pbrc{ \sum_{i=1}^{m} \pth{ X_{j-1} + X_{j+1} + 89 } } =
    O(m).
    \]
    
    The area of $K =\CHO(N)$ is bounded from below by the
    area of tiles which are not exposed by $K$. The
    probability that $K \subsetneq K_o$ (namely, the origin
    is not inside $K$, or, equivalently, all points of $N$
    lie in some semidisk) is at most $2\pi/2^n$, as easily
    verified. Hence,
    \[
    E[ Area( K ) ] \geq E[ Area( C ) ] - P \pbrc{ C \ne K
       } \pi = \pi - O(m)\frac{\pi}{n} - \frac{2\pi}{2^n}\pi
    = \pi -O \pth{n^{-2/3}}.
    \] 
    
    The assertion of the theorem now follows from Lemma
    \ref{lemma:area:to:vertices}.
\end{proof}

\begin{lemma}
    The expected number of vertices of the convex hull of
    $n$ points, chosen uniformly and independently from the
    unit square, is $O(\log{n})$.

    \label{lemma:vertices:square}
\end{lemma}

\begin{figure}
    \centerline{ \def\IPEfile{figs/sq_expose.ipe}\input{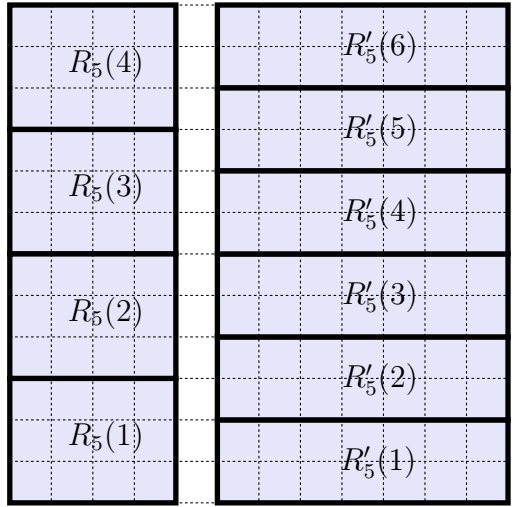} }
    \caption{Illustrating the proof that bounds the number of tiles
       exposed by $\CHO(N)$ inside the $j$-th column, by
       using a non-uniform tiling of the strips to the left
       and to the right of the $j$-th column. The area of
       such a larger tile is at least $1/n$.}
    
    \label{fig:column:ch}
\end{figure}

\begin{proof} 
    We claim that the expected area of the convex hull of
    $n$ points, chosen uniformly and independently from the
    unit square, is at least $1 - O \pth{ \log(n)/ n}$.
    
    Let $S$ denote the unit square. Partition $S$ into $n$
    rows and $n$ columns, such that $S$ is partitioned into
    $n^2$ identical squares.  Let $S_{i,j} = [(i-1)/n,i/n]
    \times [(j-1)/n,j/n]$ denote the $j$-th square in the
    $i$-th column, for $1 \leq i,j \leq n$. Let $\SC_i =
    \cup_{j=1}^{n} S_{i,j}$ denote the $i$-th column of $S$,
    for $i=1,\ldots,n$, and let $\SC(l,k) = \cup_{i=l}^{k}
    \SC_i$, for $1 \leq l \leq k \leq n$.
    
    Let $N$ be a random sample of $n$ points chosen
    uniformly and independently from $S$. Let $X_j$ denote
    the first index $i$ such that $N \cap (\cup_{l=1}^{j-1}
    S_{l, i}) \ne \emptyset$, for $j=2, \ldots, n-1$;
    namely, $X_j$ is the index of the first row in
    $\SC(1,j-1)$ that contains a point from $N$.
    Symmetrically, let $X_{j}'$ be the index of the first
    row in $\SC(j+1,n)$ that contains a point of $N$.
    Clearly, $E[X_j] = E[X_{n-j+1}']$, for $j=2,\ldots,
    n-1$.
    
    Let $Z_j$ denote the number of squares $S_{i,j}$ in the
    bottom of the $j$-th column that are exposed by
    $\CHO(N)$, for $j=2,\ldots, n-1$.  Arguing as in the
    proof of Theorem \ref{theorem:area}, we have that $Z_j
    \leq \max ( X_j, X_j' ) \leq X_j + X_j'$. Thus, in order
    to bound $E[Z_j]$, we first bound $E[X_j]$ by covering
    the strips $\SC(1,j-1), \SC(j+1,n)$ by tiles of area
    $\geq 1/n$. In particular, let $h(l) = \ceil{n/(l-1)}$,
    and let $R_j(m) = [0, (j-1)/n] \times [h(n-j+1)(m-1)/n,
    h(j)m/n ]$, and let $R_j'(m) = [(j+1)/n, 1] \times
    [h(j)(m-1)/n, h(j)m/n ]$, for $j=2, \ldots, n-1$.  See
    Figure \ref{fig:column:ch}.
    
    Let $Y_j$ denote the minimal index $i$ such that $R_j(i)
    \cap N \ne \emptyset$.  The area of $R_{j}(i)$ is at
    least $1/n$, for any $i$ and $j$. Arguing as in the
    proof of Theorem \ref{theorem:area}, it follows that $E[
    Y_j ] = O(1)$. On the other hand, $E[ X_j] \leq h(j)
    E[Y_j] = O( n/(j-1))$.  Symmetrically, $E[X_j'] =
    O(n/(n-j))$.
    
    Thus, by applying the above argument to the four
    directions (top, bottom, left, right), we have that the
    expected number of squares $S_{i,j}$ exposed by
    $\CHO(N)$ is bounded by
    \[
    4n - 4 + 4 { \sum_{j=2}^{n-1} E[Z_j]} < 4n + 4 {
       \sum_{j=2}^{n-1} (E[X_j] + E[X_j'])} = 4n + 8 {
       \sum_{j=2}^{n-1} O\pth{\frac{n}{j-1}} } =
    O(n\log{n}),
    \]
    where $4n - 4$ is the number of squares adjacent to the
    boundary of $S$.
    
    Since the area of each square is $1/n^2$, it follows
    that the expected area of $\CHO(N)$ is at least $1 -
    O(\log(n)/n)$.
    
    By Lemma \ref{lemma:area:to:vertices}, the expected
    number of vertices of the convex hull is $O(\log n)$.
\end{proof}

\begin{lemma}
    The expected number of vertices of the convex hull of
    $n$ points, chosen uniformly and independently from a
    triangle, is $O(\log{n})$.
    
    \label{lemma:triangle}
\end{lemma}

\begin{proof}
    We claim that the expected area of the convex hull of
    $n$ points, chosen uniformly and independently from a
    triangle $T$, is at least $(1 - O \pth{ \log(n)/ n})
    Area(T)$. We adapt the tiling used in Lemma
    \ref{lemma:vertices:square} to a triangle. Namely, we
    partition $T$ into $n$ equal-area triangles, by segments
    emanating from a fixed vertex, each of which is then
    partitioned into $n$ equal-area trapezoids by segments
    parallel to the opposite side, such that each resulting
    trapezoid has area $1/n^2$. See Figure
    \ref{fig:triangle:tiling}.
    
    Notice that this tiling has identical topology to the
    tiling used in Lemma \ref{lemma:vertices:square}. Thus,
    the proof of Lemma \ref{lemma:vertices:square} can be
    applied directly to this case, repeating the tiling
    process three times, once for each vertex of $T$. This
    readily implies the asserted bound.
\end{proof}

\begin{figure}
    \centerline{
       \def\IPEfile{figs/tri_expose.ipe}\input{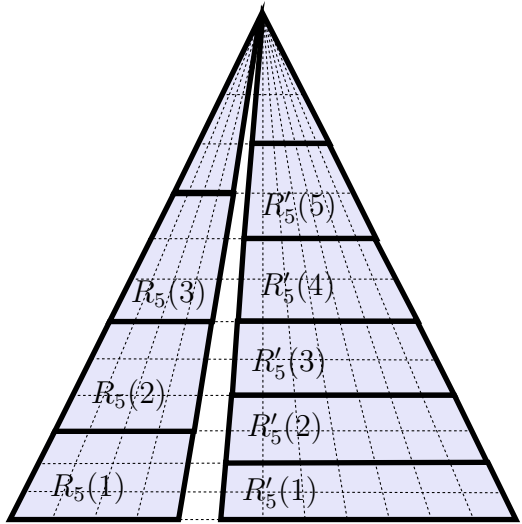}
       }
%           \Ipe{triangle_grid.ipe}
    
    \caption{Illustrating the proof of Lemma 
       \ref{lemma:vertices:square} for the case of a
       triangle.}
    
    \label{fig:triangle:tiling}
\end{figure}

\begin{theorem}
    The expected number of vertices of the convex hull of
    $n$ points, chosen uniformly and independently from a
    polygon $P$ having $k$ sides, is $O(k \log{n})$.
\end{theorem}

\begin{proof}
    We triangulate $P$ in an arbitrary manner into $k$
    triangles $T_1, \ldots, T_k$. Let $N$ be a random sample
    of $n$ points, chosen uniformly and independently from
    $P$. Let $Y_i = |T_i \cap N|$, $N_i = T_i \cap N$, and
    $Z_i = |\CHO(N_i)|$, for $i=1, \ldots, k$. Notice that
    the distribution of the points of $N_i$ inside $T_i$ is
    identical to the distribution of $Y_i$ points chosen
    uniformly and independently from $T_i$. In particular,
    $E[Z_i | Y_i] = O( \log{Y_i})$, by Lemma
    \ref{lemma:triangle}, and $E[ Z_i ] = E_{Y_i}[ E[ Z_i |
    Y_i ] ] = O( \log{n } )$, for $i=1, \ldots, k$.
    
    Thus, $E[ |\CHO(N)| ] \leq E \pbrc{ \sum_{i=1}^{k}
       |\CHO(N_i)| } \leq \sum_{i=1}^{k} E[ Z_i] = O( k
    \log{n} )$.
\end{proof}

\section{On the Expected Complexity of a
   Generalized Convex Hull Inside a Disk}
\label{sec:genrelized:convex}

In this section, we derive a bound on the expected
complexity on a generalized convex hull of a set of points,
chosen uniformly and independently for the unit disk. The
new bound matches the known bounds, for the case of standard
convexity and maxima. 
The bound follows by extending the proof of  Theorem
\ref{theorem:area}.

We begin with some terminology and some initial
observations, most of them taken or adapted from
\cite{mp-ofsch-97}.  A set $\D$ of vectors in the plane is a
{\em set of directions}, if the length of all the vectors in
$\D$ is $1$, and if $v \in \D$ then $-v \in \D$. Let
$\D_{\Re}$ denote the set of all possible directions.  A set
$C$ is {\em $\D$-convex} if the intersection of $C$ with any
line with a direction in $\D$ is connected.  By definition,
a set $C$ is convex (in the standard sense), if and only if
it is $\D_{\Re}$-convex.

For a set $C$ in the plane, we denote by $\CHX_{D}(C)$ the
{\em $\D$-convex hull} of $C$; that is, the smallest
$\D$-convex set that contains $C$.  While this seems like a
reasonable extension of the regular notion of convexity, its
behavior is counterintuitive. For example, let $\D_Q$ denote
the set of all rational directions (the slopes of the
directions are rational numbers). Since $\D_Q$ is dense in
$\D_{\Re}$, one would expect that $\CHX_{\D_Q}(C) =
\CHX_{\D_\Re}(C) = \CH(C)$. However, if $C$ is a set of
points such that the slope of any line connecting a pair of
points of $C$ is irrational, then $\CHX_{D_Q}(C) =C$.  See
\cite{osw-dcrch-85, rw-cgro-88, rw-ocfoc-87} for further
discussion of this type of convexity.

\begin{definition}
    Let $f$ be a real function defined on a $\D$-convex set
    $C$. We say that $f$ is {\em $\D$-convex} if, for any $x
    \in C$ and any $v \in \D$, the function $g(t)=f(x+tv)$
    is a convex function of the real variable $t$. (The
    domain of $g$ is an interval in $\Re$, as $C$ is assumed
    to be $\D$-convex.)
   
    Clearly, any convex function, in the standard sense,
    defined over the whole plane satisfies this condition.
\end{definition}

\begin{definition}
    Let $C \subseteq \Re^2$. The set $\CHX^\D(C)$, called
    the {\em functional $\D$-convex hull of $C$}, is defined
    as
    \[
    \CHX^\D(C) = \brc{ x\in \Re^2 \sep{ f(x) \leq \sup_{y\in
             C}f(y) \text{ for all $\D$-convex } f:\Re^2
          \rightarrow \Re}}
    \]
    A set $C$ is {\em functionally $\D$-convex} if
    $C=\CHX^{\D}(C)$.
\end{definition}

\begin{definition}
    Let $\D$ be a set of directions. A pair of vectors
    $v_1,v_2 \in \D$, is a {\em $\dpair$}, if $v_2$ is
    counterclockwise from $v_1$, and there is no vector in
    $\D$ between $v_1$ and $v_2$. Let $\DPAIRS{\D}$ denote
    the set of all $\dpair$s.  Let $\pspan(u_1,u_2)$ denote
    the portion of the plane that can be represented as a
    {\em positive} linear combination of $u_1, u_2 \in \D$.
    Thus $\pspan( u_1,u_2)$ is the {\em open} wedge bounded
    by the rays emanating from the origin in directions
    $u_1, u_2$. We define by $(v_1,v_2)_L = \pspan( -v_1,
    v_2)$ and $(v_1,v_2)_R = \pspan( v_1, -v_2)$: these are
    two of the four quadrants of the plane induced by the
    lines containing $v_1$ and $v_2$. Similarly, for $v \in
    \D$ we denote by $\HL{v}$ and $\HR{v}$ the two open
    half-planes defined by the line passing through $v$.
    Let
   \[
   \QD(\D) = \brc{ \HL{v}, \HR{v} \sep{ v \in \D}} \cup
   \brc{ (v_1,v_2)_R, (v_1,v_2)_L \sep{ (v_1, v_2) \in
         \DPAIRS{D} }}.
   \]
\end{definition}

\begin{definition}
    For a set $S \subseteq \Re^2$ we denote by $T(S)$ the
    set of translations of $S$ in the plane, that is $T(S) =
    \brc{ S + p \sep{ p \in \Re^2 }}$. Given a set of
    directions $\D$, let $\T(\D) = \bigcup_{Q \in \QD(\D)}
    T(Q)$.  
\end{definition}

For $\D_\Re$, the set $\T(\D_\Re)$ is the set of all open
half-planes.  The standard convex hull of a planar point set
$S$ can be defined as follows: start from the whole plane,
and remove from it all the open half-planes $H^+$ such that
$H^+ \cap S = \emptyset$.  We extend this definition to
handle $\D$-convexity for an arbitrary set of directions
$\D$, as follows:
\[
\DCH{\D}(S) = \Re^2 \setminus \pth{ \bigcup_{I \in \T(\D), I
      \cap S = \emptyset} I };
\]
that is, we remove from the plane all the translations of
quadrants and halfplanes in $\QD(\D)$ that do not contain a
point of $S$. See Figures \ref{fig:dch:example},
\ref{fig:dch:example:ext}.

\begin{figure}
    \centerline{ \def\IPEfile{figs/set-of-dir.ipe}\input{figs/set-of-dir.ipe} }
    \caption{(a) A set of directions $\D$, (b) the set of
       quadrants $\QD(\D)$ induced by $\D$, and (c) the
       $\DCH{\D}$ of three points.}
    
    \label{fig:dch:example}
\end{figure}

\begin{figure}
    \centerline{ \def\IPEfile{figs/set-of-dir-discon.ipe}\input{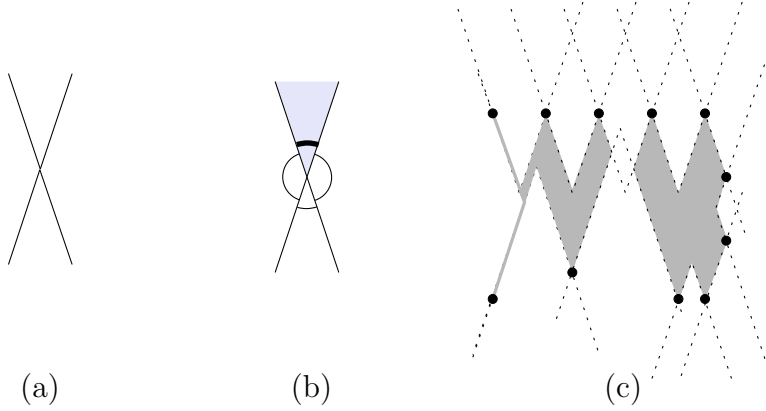} }
    \caption{(a) A set of directions $\D$, such that $\alpha(\D) > \pi/2$,
       (b) the set of quadrants $\QD(\D)$ induced by $\D$,
       and (c) the $\DCH{\D}$ of a set of points which is
       not connected.}
    
    \label{fig:dch:example:ext}
\end{figure}

For the case $\D_{xy} = \brc{ (0,1), (1,0), (0,-1), (-1,0)
   }$, \matousek{} and \plechac{} \cite{mp-ofsch-97} showed
that if $\DCH{\D_{xy}}(S)$ is connected, then
$\CHX^{\D_{xy}}(S) = \DCH{\D_{xy}}(S)$.

\begin{definition}
    For a set of directions $\D$, we define the
    {\em density} of $\D$ to be
    \[
    \alpha(\D) = \max_{(v_1,v_2) \in \DPAIRS{\D}}
    \alpha(v_1,v_2),
    \]
    where $\alpha(v_1,v_2)$ denotes the counterclockwise
    angle from $v_1$ to $v_2$.
\end{definition}

See Figure \ref{fig:dch:example:ext}, for an example of a
set of directions with density larger than $\pi/2$.

\begin{corollary}
    Let $\D$ be a set of directions in the plane. Then:
    \begin{itemize}
        \item The set $\DCH{\D}(A)$ is $\D$-convex, for any
        $A \subseteq \Re^2$.
        
        \item For any $A \subseteq B \subseteq \Re^2$, one
        has $\DCH{\D}(A) \subseteq \DCH{\D}(B)$.
        
        \item For two sets of directions $\D_1 \subseteq
        \D_2$ we have $\DCH{\D_1}(S) \subseteq
        \DCH{\D_2}(S)$, for any $S \subseteq \Re^2$.
        
        \item Let $S$ be a bounded set in the plane, and let
        $\D_1 \subseteq \D_2 \subseteq \D_3 \cdots$ be a
        sequence of sets of directions, such that
        $\lim_{i\rightarrow \infty} \alpha(D_i) = 0$. Then,
        $\inter{\CH(S)} \subseteq \lim_{i \rightarrow
           \infty} \DCH{\D_i}(S) \subseteq \CH(S)$.
        \remove{
           \item Let $S$ be a finite set of points in general
           position in the plane, and let $\D_1 \supseteq \D_2
           \supseteq \D_3 \cdots$ be a sequence of sets of
           directions, such that $\lim_{i\rightarrow \infty}
           \alpha(D_i) = \pi$. Then, $\lim_{i \rightarrow
              \infty} \DCH{\D_i}(S) = S$.}
    \end{itemize}
\end{corollary}

\begin{lemma}
    Let $\D$ a set of directions, and let $S$ be a finite
    set of points in the plane. Then $C = \DCH{\D}(S)$ is a
    polygonal set whose complexity is $O(|S \cap \bd{C}|)$.

    \label{lemma:on:boundary}
\end{lemma}

\begin{proof}
    It is easy to show that $C$ is polygonal.  We charge
    each vertex of $C$ to some point of $S' = S \cap\bd{C}$.
    Let $C'$ be a connected component of $C$. If $C'$ is a
    single point, then this is a point of $S'$. Otherwise,
    let $e$ be an edge of $C'$, and let $I$ be a set in
    $\T(\D)$ such that $e \subseteq \bd{I}$, and $I \cap S =
    \emptyset$.
    
    Since $e$ is an edge of $C'$, there is no $q \in \Re^2$
    such that $e \subseteq q + I$, and $(q+I)\cap S =
    \emptyset$.  This implies that there must be a point $p$
    of $S$ on $\bd{I} \cap l_e$, where $l_e$ is the line
    passing through $e$. However, $C$ is a $\D$-convex set,
    and the direction of $e$ belongs to $\D$.  It follows
    that $l_e$ intersects $C$ along a connected set (i.e.,
    the segment $e$), and $p \in l_e \cap C = e$. We charge
    the edge $e$ to $p$.  We claim that a point $p$ of $S'$
    can be charged at most 4 times.  Indeed, for each edge
    $e'$ of $C$ incident to $p$, there is a supporting set
    in $\T(\D)$, such that $p$ and $e'$ lie on its boundary.
    Only two of those sets can have angle less than $\pi/2$
    at $p$ (because such a set corresponds to a
    $\dpair(v_1,v_2)$ with $\alpha(v_1,v_2) > \pi/2$). Thus,
    a point of $S'$ is charged at most $\max( 2\pi/(\pi/2),
    \pi/(\pi/2) + 2) = 4$ times.
\end{proof}

%\subse ction{On the expected complexity of the generalized convex  
%   hull of a random point set inside a disk}
%\label{sec:expected:complexity}

%The following Lemma extends Lemma
%\ref{lemma:area:to:vertices}.

\begin{lemma}
    Let $\D$ be a set of directions, and let $K$ be a
    bounded convex body in the plane, such that the expected
    area of $\DCH{\D}(N)$ of a set $N$ of $n$ points, chosen
    uniformly and independently from $K$, is at least
    $\pth{1-f(n)}Area(K)$, where $1 \geq f(n) \geq 0$, for
    $n \geq 1$.  Then, the expected number of vertices of $C
    = \DCH{\D}(N)$ is $O(n f(n/2))$.

    \label{lemma:area:to:vertices:ext}
\end{lemma}

\begin{proof}
    By Lemma \ref{lemma:on:boundary}, the complexity of $C$
    is proportional to the number of points of $N$ on the
    boundary of $C$. Using this observation, it is easy to
    verify that the proof of Lemma
    \ref{lemma:area:to:vertices} can be extended to this
    case.
\end{proof}

We would like to apply the proof of Theorem 2.3 to bound the
expected complexity of a random $\D$-convex hull inside a
disk. Unfortunately, if we try to concentrate only on three
consecutive sectors (as in Figure \ref{fig:slice}) it might
be that there is a quadrant $I$ of $\T(\D)$ that intersects
the middle the middle sector from the side (i.e. through the
two adjacent sectors). This, of course, can not happen when
working with the regular convexity. Thus, we first would
like to decompose the unit disk into ``safe'' regions, where
we can apply a similar analysis as the regular case, and the
``unsafe'' areas. To do so, we will first show that, with
high probability, the $\DCH{\D}$ of a random point set
inside a disk, contains a ``large'' disk in its interior.
Next, we argue that this implies that the random $\DCH{\D}$
covers almost the whole disk, and the desired bound will
readily follows from the above Lemma.

\begin{definition}
    For $r \geq 0$, let $B_r$ denote the disk of
    radius of $r$ centered at the origin.
\end{definition}

\begin{lemma}
    Let $\D$ be a set of directions, such that $0 \leq
    \alpha(\D) \leq \pi / 2$. Let $N$ be a set of $n$ points
    chosen uniformly and independently from the unit disk.
    Then, with probability $1-n^{-10}$ the set $\DCH{\D}(N)$
    contains $B_r$ in its interior, where $r = 1 - c
    \sqrt{\log{n}/n}$, for an appropriate constant $c$.

    \label{lemma:big:disk}
\end{lemma}

\begin{proof}
    Let $r'=1 - c \sqrt{(\log{n})/n}$, where $c$ is a
    constant to be specified shortly. Let $q$ be any point
    of $B_{r'}$. We bound the probability that $q$ lies
    outside $C = \DCH{\D}(N)$ as follows: Draw $8$ rays
    around $q$, such that the angle between any two
    consecutive rays is $\pi/4$. This partitions $q +
    B_{r''}$, where $r'' = c \sqrt{(\log{n})/n}$, into eight
    portions $R_1, \ldots, R_8$, each having area $\pi c^2
    \log{n}/(8n)$. Moreover, $R_i \subseteq q + B_{r''}
    \subseteq B_1$, for $i=1, \ldots, 8$. The probability of
    a point of $N$ to lie outside $R_i$ is $1 - c^2
    \log{n}/(8n)$. Thus, the probability that all the points
    of $N$ lie outside $R_i$ is
    \[
    P \pbrc{ N \cap R_i = \emptyset } \leq \pth{ 1 - \frac{
          c^2 \log{n}}{8n}}^n \leq e^{-(c^2\log{n})/8} =
    n^{-c^2/8},
    \]
    since $1-x \leq e^{-x}$, for $x \geq 0$. Thus, the
    probability that one of the $R_i$'s does not contain a
    point of $N$ is bounded by $8 n^{-c^2/8}$. We claim that
    if $R_i \cap N \ne \emptyset$, for every $i=1,\ldots,
    8$, then $q \in C$.  Indeed, if $q \notin C$ then there
    exists a set $Q \in \QD(\D)$, such that $(q + Q) \cap N
    = \emptyset$. Since $\alpha(\D) \leq \pi/2$ there exists
    an $i$, $1 \leq i \leq 8$, such that $R_i \subseteq q +
    Q$; see Figure \ref{fig:quadrants}. This is a
    contradiction, since $R_i \cap N \ne \emptyset$.  Thus,
    the probability that $q$ lies outside $C$ is $\leq
    8n^{-c^2/8}$.

    \begin{figure}
        \centerline{
           \def\IPEfile{figs/quadrants.ipe}\input{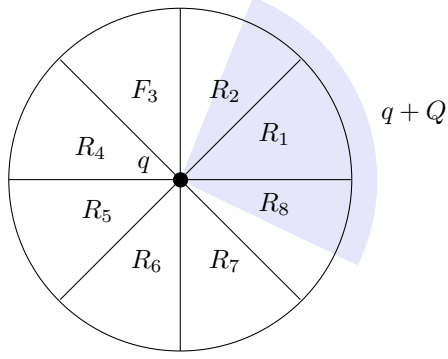}
           }
        \caption{Since $\alpha(\D) \leq \pi/2$, any quadrant 
           $Q \in \QD(\D)$, when translated by $q$, must
           contain one of the $R_i$'s.}
        
        \label{fig:quadrants}
    \end{figure}

    Let $N'$ denote a set of $n^{10}$ points spread
    uniformly on the boundary of $B_{r'}$. By the above
    analysis, all the points of $N'$ lie inside $C$ with
    probability at least $1-8n^{10-c^2/8}$.  Furthermore,
    arguing as above, we conclude that $B_{r} \subseteq
    \DCH{\D}(N')$, where $r = 1 - 2 c \sqrt{(\log{n})/n}$.
    Hence, with probability at least $1 - 8n^{10-c^2/8}$,
    $\DCH{\D}(C)$ contains $B_r$.  The lemma now follows by
    setting $c=20$, say.
\end{proof}

Since the set of directions may contain large gaps, there
are points in $B_1 \setminus B_r$ that are ``unsafe'', in
the following sense:

\begin{definition}
    Let $\D$ be a set of directions, and let $0 \leq r \leq
    1$ be a prescribed constant, such that $0 \leq
    \alpha(\D) \leq \pi / 2$. A point $p$ in $B_1$ is {\em
       safe}, relative to $B_r$, if $op \subseteq
    \DCH{\D}(B_r \cup \brc{p})$.
\end{definition}

See Figure \ref{fig:unsafe} for an example how the unsafe
area looks like. The behavior of the $\DCH{\D}$ inside the
unsafe areas is somewhat unpredictable. Fortunately, those
areas are relatively small.

\begin{figure}
    \centerline{
       \def\IPEfile{figs/unsafe1.ipe}\input{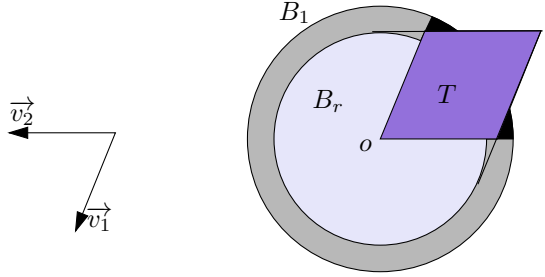}
       }
    \caption{The dark areas are the unsafe areas for a 
       consecutive pairs of directions
       $v_1, v_2 \in \D$.}
    
    \label{fig:unsafe}
\end{figure}

\begin{lemma}
    Let $\D$ be a set of directions, such that $0 \leq
    \alpha(\D) \leq \pi / 2$, and let $r= 1 - O \pth{
       \sqrt{(\log{n})/n}}$. The unsafe area in $B_1$,
    relative to $B_r$, can be covered by a union of $O(1)$ caps.
    Furthermore, the length of the base of such a cap is
    $O(((\log{n})/n)^{1/4})$, and its height is
    $O\pth{\sqrt{(\log{n})/n}}$.

    \label{lamma:bad:bad:bad:caps}
\end{lemma}

\begin{proof}
    Let $p$ be an unsafe point of $B_1$. Let
    $\overrightarrow{v_1},\overrightarrow{v_2}$ be the
    consecutive pair of vectors in $\D$, such that the
    vector $\overrightarrow{po}$ lies between them. If
    $\ray(p,\overrightarrow{v_1}) \cap B_r \ne \emptyset$,
    and $\ray(p,\overrightarrow{v_2}) \cap B_r \ne
    \emptyset$ then $po \subseteq \CH \pth{ \brc{ p, o, p_1,
          p_2 } } \subseteq \DCH{\D}(B_r \cup \brc{p})$, for
    any pair of points $p_1 \in B_r \cap
    \ray(p,\overrightarrow{v_1}), p_2 \in B_r \cap
    \ray(p,\overrightarrow{v_2})$.  Thus, $p$ is unsafe only
    if one of those two rays miss $B_r$. Since $p$ is close
    to $B_r$, the angle between the two tangents to $B_r$
    emanating from $p$ is close to $\pi$.  This implies that
    the angle between $\overrightarrow{v_1}$ and
    $\overrightarrow{v_2}$ is at least $\pi/4$ (provided $n$
    is a at least some sufficiently large constant), and the
    number of such pairs is at most $8$.
    
    The area in the plane that sees $o$ in a direction
    between $\overrightarrow{v_1}$ and
    $\overrightarrow{v_2}$, is a quadrant $Q$ of the plane.
    The area in $Q$ which is is safe, is a parallelogram
    $T$.  Thus, the unsafe area in $B_1$ that induced by the
    pair $\overrightarrow{v_1}$ and $\overrightarrow{v_2}$
    is $(B_1 \cap Q) \setminus T$. Since $\alpha(\D) \leq
    \pi/2$, this set can covered with two caps of $B_1$ with
    their base lying on the boundary of $B_r$. See Figure
    \ref{fig:unsafe}.
    
    The height of such a cap is $1-r =
    O\pth{\sqrt{\frac{\log{n}}{n(\pi - \alpha)}}}$, and the
    length of the base of such a cap is $2\sqrt{ 1- r^2 } =
    O\pth{ \pth{\frac{\log{n}}{n(\pi-\alpha)}}^{1/4} }$.
\end{proof}

The proof of Lemma \ref{lamma:bad:bad:bad:caps} is where our
assumption that $\alpha(\D) \leq \pi/2$ plays a critical
role. Indeed, if $\alpha(\D) > \pi/2$, then the unsafe areas
in $B_1 \setminus B_r$ becomes much larger, as indicated by
the proof.

\begin{theorem}
    Let $\D$ be a set of directions, such that $0 \leq
    \alpha(\D) \leq \pi / 2$.  The expected number of
    vertices of $\DCH{\D}(N)$, where $N$ is a set of $n$
    points, chosen uniformly and independently from the unit
    disk, is $O\pth{n^{1/3} + \sqrt{n\alpha(\D)}}$.
    
    \label{theorem:area:x}
\end{theorem}

\begin{proof}
    We claim that the expected area $\DCH{\D}(N)$ is at
    least $\pi - O \pth{ n^{-2/3} + \sqrt{\alpha/n} }$,
    where $\alpha = \alpha(\D)$.  The theorem will then
    follow from Lemma \ref{lemma:area:to:vertices:ext}.
    
    Indeed, let $m$ be an integer to be specified later, and
    assume, without loss of generality, that $m$ divides
    $n$.  Partition $B$ into $m$ congruent sectors, $\SC_1,
    \ldots, \SC_m$.  Let $B^1, \ldots, B^{\mu}$ denote the
    $\mu = n/m$ disks centered at the origin, such that (i)
    $B^{1} = B_1$, and (ii) $Area(B^{i-1}) - Area(B^{i}) =
    \pi/\mu$, for $i=2, \ldots, \mu$.  Let $r_i$ denote the
    radius of $B^i$, for $i=1, \ldots, \mu$.  Note\footnote{
       $Area(B^1) - Area(B^2) = \pi(1-r_2^2) = \pi/\mu$,
       thus $r_2^2 = 1 - 1/\mu$. We have $r_2 \leq 1
       -1/(2\mu)$, and $r_1 - r_2 \geq 1 - (1-1/(2\mu)) =
       1/(2\mu)$.}, that $r_{i} - r_{i+1} \geq r_{i-1} -
    r_{i}\geq 1/(2\mu)$, for $i=2, \ldots, \mu-1$.
    
    Let $r= 1 - O\pth{\sqrt{(\log{n})/n}}$, and let $U$ be
    the set of sectors that either intersect an unsafe area
    of $B$ relative to $B_r$, or their neighboring sectors
    intersect the unsafe area of $B$. By Lemma
    \ref{lamma:bad:bad:bad:caps}, the number of sectors in
    $U$ is $O(1) \cdot O \pth{\frac{ ((\log{n})/n)^{1/4}}
       {(2\pi/m)}} = O(m((\log{n})/n)^{1/4})$.
    
    Let $S_{i,j} = (B^{i} \setminus B^{i+1}) \cap \SC_j$,
    and $S_{\mu,j}=B^{\mu} \cap \SC_j$, for $i=1, \ldots,
    \mu-1$, and $j=1, \ldots, m$. The set $S_{i,j}$ is
    called the $i$-th {\em tile} of the sector $\SC_j$, and
    its area is $\pi/n$, for $i=1,\ldots,\mu$, and $j=1,
    \ldots,m$.
    
    Let $X_j$ denote the first index $i$ such that $N \cap
    S_{i, j} \ne \emptyset$, for $j=1, \ldots, m$. The
    probability that $X_j = k$ is upper-bounded by the
    probability that the tiles $S_{1, j}, \ldots,
    S_{(k-1),j}$ do not contain any point of $N$; namely, by
    $\pth{1-\frac{k-1}{n}}^{n}$.  Thus, $P[X_j = k] \leq
    \pth{1-\frac{k-1}{n}}^{n} \leq e^{-(k-1)}$. Thus,
    \[
    E \pbrc{ X_j } = \sum_{k=1}^{\mu} k P[X_j = k ] \leq
    \sum_{k=1}^{\mu} k e^{-(k-1)} = O(1),
    \]
    for $j=1, \ldots, m$.
    
    Let $C$ denote the set $\DCH{\D}(N \cup B_r)$. The tile
    $S_{i,j}$ is {\em exposed} by a set $K$, if $S_{i,j}
    \setminus K \ne \emptyset$.
    
    We claim that the expected number of tiles exposed by
    $C$ in a section $S_j \notin U$ is at most $X_{j-1} +
    X_{j+1} + O(\mu/m^2 + \alpha\mu/m)$, for $j=1,\ldots, m$
    (where we put $X_0 = X_m$, $X_{m+1} = X_1$).
    
%    \begin{figure}
%        \centerline{
%           \Ipe{figs/expose.ipe}
%           }
%        \caption{Illustrating the proof that bounds the number of tiles
%           exposed by $T$ inside $\SC_j$}
    
%        \label{fig:slice}
%    \end{figure}
    
    Indeed, let $w=max(X_{j-1},X_{j+1})$, and let $p,q$ be
    the two points in $S_{j-1,w}, S_{j+1,w}$, respectively,
    such that the number of sets exposed by the triangle $T
    = \triangle{opq}$, in the sector $\SC_j$, is maximal.
    Both $p$ and $q$ lie on $\bd{B^{w+1}}$ and on the
    external radii bounding $\SC_{j-1}$ and $\SC_{j+1}$, as
    shown in Figure \ref{fig:slice}.  Let $s$ denote the
    segment connecting the midpoint $\rho$ of the base of
    $T$ to its closest point on $\bd{B^w}$. The number of
    tiles in $\SC_j$ exposed by $T$ is bounded by $w$, plus
    the number of tiles intersecting the segment $s$.  The
    length of $s$ is
    \[
    |oq| - |oq| \cos \pth{ \frac{3}{2} \cdot \frac{2\pi}{m}}
    \leq 1 - \cos \pth{ \frac{3\pi}{m}} \leq \frac{1}{2}
    \pth{ \frac{3\pi}{m}}^2 = \frac{4.5\pi^2}{m^2},
    \]
    since $\cos{x} \geq 1-x^2/2$, for $x \geq 0$.
    
    On the other hand, the segment $s$ intersects at most
    $\ceil{||s||/(1/(2\mu))} = O(\mu/m^2)$ tiles, and we
    have that the number of tiles exposed in the sector
    $\SC_i$ by $T$ is at most $w + O(\mu/m^2)$, for $j=1,
    \ldots, m$.
    
    Since $\SC_j \notin U$, the points $p,q$ are safe, and
    $op, oq \subseteq C$. This implies that the only
    additional tiles that might be exposed in $\SC_j$ by
    $C$, are exposed by the portion of the boundary of $C$
    between $p$ and $q$ that lie inside $T$.  Let $V$ be the
    circular cap consisting of the points in $T$ lying
    between $pq$ and a circular arc $\gamma \subseteq T$,
    connecting $p$ to $q$, such that for any point $p' \in
    \gamma$ one has $\angle{pp'q} = \pi - \alpha$. See
    Figure \ref{fig:circ:arc}.

    \begin{figure}
        \centerline{ 
           \def\IPEfile{figs/attack.ipe}\input{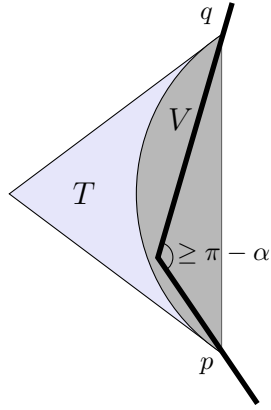} }
        \caption{The portion of $T$ that can be removed by a quadrant
           $Q$ of $\T(\D)$, is covered by the darkly-shaded
           circular cap, such that any point on its bounding
           arc creates an angle $\pi-\alpha$ with $p$ and
           $q$.}
        
        \label{fig:circ:arc}
    \end{figure}
    
    Let $Q \in \T(\D)$ be any quadrant of the plane induce
    by $\D$, such that $Q \cap N = \emptyset$ (i.e. $C\cap Q
    = \emptyset$), and $Q \cap T \ne \emptyset$.  Then, $Q
    \cap op = \emptyset, Q \cap oq =\emptyset$ since $p$ and
    $q$ are safe. Moreover, the angle of $Q$ is at least
    $\pi - \alpha$, which implies that $Q \cap T \subseteq
    V$.  See Figure \ref{fig:circ:arc}.
    
    Let $s'$ be the segment $o\rho \cap V$, where $\rho$ is
    as above, the midpoint of $pq$. The length of $s''$ is
    \[
    |s'| \leq \sin \pth{ \frac{3}{2} \cdot \frac{2\pi}{m}}
    \tan{ \frac{\alpha}{2} } \leq \frac{3\pi}{m}
    \frac{\sqrt{2}\alpha}{2} \leq \frac{3\pi\alpha}{m},
    \]
    since $\sin{x} \leq x$, for $x \geq 0$, and $1/\sqrt{2}
    \leq \cos{(\alpha/2)}$ (because $0 \leq \alpha \leq
    \pi/2$).
    
    Thus, the expected number of tiles exposed by $C$, in a
    sector $\SC_j \notin U$, is bounded by
    \[
    X_{j-1} + X_{j+1} + O \pth{ \frac{\mu}{m^2}} + O \pth{
       \frac{3\pi\alpha/m}{1/(2\mu)} } = X_{j-1} + X_{j+1} +
    O \pth{ \frac{\mu}{m^2}} + O \pth{ \frac{\alpha
          \mu}{m}}.
    \]
    
    Thus, the expected number of tiles exposed by $C$, in
    sectors that do not belong to $U$, is at most
    \[
    E \pbrc{ \sum_{j=1}^{m} \pth{ X_{j-1} + X_{j+1} + O
          \pth{ \frac{\mu}{m^2}} + O \pth{ \frac{\alpha
                \mu}{m}} } } = O \pth{ m + \frac{\mu}{m} +
       \alpha \mu}.
    \]
    
    Adding all the tiles that lie outside $B_r$ in the
    sectors that belong to $U$, it follows that the expected
    number of tiles exposed by $C$ is at most
    \begin{eqnarray*}
        O&&\hspace{-0.75cm} \pth{ m + \frac{\mu}{m} + \alpha
           \mu + |U|\cdot \frac{1-r}{1/2\mu} }= O \pth{ m +
           \frac{\mu}{m} + \alpha \mu + m
           \pth{\frac{\log{n}}{n}}^{1/4} \cdot
           \mu \sqrt{\pth{\frac{\log{n}}{n}}}} \\
        &=& O \pth{ m + \frac{n}{m^2} + \frac{\alpha n}{m} +
           n\pth{\frac{\log{n}}{n}}^{3/4} } = O \pth{ m +
           \frac{n}{m^2} + \frac{\alpha n}{m} +
           n^{1/4}\log^{3/4}{n} }.
    \end{eqnarray*}
    Setting $m=\max{\pth{n^{1/3}, \sqrt{n\alpha}}}$, we
    conclude that the expected number of tiles exposed by
    $C$ is $O\pth{n^{1/3} + \sqrt{n\alpha}}$.
    
    The area of $C' =\DCH{\D}(N)$ is bounded from below by
    the area of the tiles which are not exposed by $C'$. The
    probability that $C' \ne C$ (namely, that the disk $B_r$
    is not inside $C'$) is at most $n^{-10}$, by Lemma
    \ref{lemma:big:disk}.  Hence the expected area of $C'$
    is at least
    \[
    E[ Area( C ) ] - Prob \pbrc{ C \ne C' } \pi = \pi -
    O\pth{ n^{1/3} + \sqrt{n\alpha}}\frac{\pi}{n} -
    n^{-10}\pi = \pi -O \pth{n^{-2/3} +
       \sqrt{\frac{\alpha}{n}} \;}.
    \] 
    
    The assertion of the theorem now follows from Lemma
    \ref{lemma:area:to:vertices:ext}.
\end{proof}

The expected complexity of the $\DCH{\D_{xy}}$ of $n$
points, chosen uniformly and independently from the unit
square, is $O(\log{n})$ (Lemma \ref{lemma:vertices:square}).
Unfortunately, this is a degenerate case for a set of
directions with $\alpha(\D) = \pi/2$, as the following
corollary testifies:

\begin{corollary}
    Let $\D_{xy}'$ be the set of directions resulting from
    rotating $\D_{xy}$ by 45 degrees. Let $N$ be a set of
    $n$ points, chosen independently and uniformly from the
    unit square $\SQR$. The expected complexity of
    $\DCH{\D_{xy}'}(N)$ is $\Omega \pth{\sqrt{n}}$.
\end{corollary}

\begin{proof}
    Without loss of generality, assume that $n=m^2$ for some
    integer $m$.  Tile $\SQR$ with $n$ translated copies of
    a square of area $1/n$.  Let $\S_1, \ldots, \S_m$ denote
    the squares in the top raw of this tiling, from left to
    right. Let $A_j$ denote the event that $\S_j$ contains a
    point of $N$, and neither of the two adjacent squares
    $S_{j-1}, S_{j+1}$ contains a point of $N$, for $j=2,
    \ldots, m-1$.

    We have
    \[
    Prob\pbrc{A_j} = Prob \pbrc{ \S_{j+1} \cap N = \emptyset
       \text{ and }\S_{j-1} \cap N = \emptyset } - Prob
    \pbrc{ (\S_{j-1} \cup \S_j \cup \S_{j+1}) \cap N =
       \emptyset },
    \]
    for $j=2, \ldots, m-1$. Hence,
    \[
    \lim_{n\rightarrow \infty} Prob\pbrc{A_j} =
    \lim_{n\rightarrow \infty } \pth{ \pth{1 -
          \frac{2}{n}}^{n} - \pth{1 - \frac{3}{n}}^{n} } =
    e^{-2} - e^{-3} \approx 0.0855
    \]
    
    \begin{figure}
        \centerline{
           \def\IPEfile{figs/squares.ipe}\input{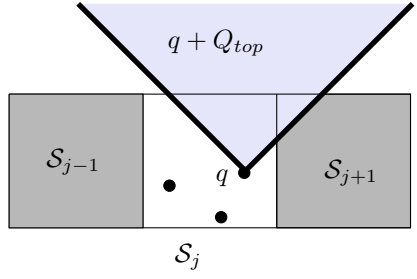}
           }
        \caption{If $A_j$ happens, then the squares
           $\S_{j-1}, \S_{j+1}$ do not
           contain a point of $N$. Thus, if $q$ is the
           highest point in $\S_j$, then $q+Q_{top}$ can not
           contain a point of $N$, and $q$ is a vertex of
           $\DCH{\D_{xy}'}(N)$.}
        
        \label{fig:squares}
    \end{figure}
    
    This implies, that for $n$ large enough, $Prob
    \pbrc{A_j} \geq 0.01$.  Thus, the expected value of $Y$
    is $\Omega(m) = \Omega \pth{ \sqrt{n}}$, where $Y$ is
    the number of $A_j$'s that have occurred, for
    $j=2,\ldots, m-1$.  However, if $A_j$ occurs, then $C
    =\DCH{\D_{xy}'}(N)$ must have a vertex at $\S_j$.
    Indeed, let $Q_{top}$ denote the quadrant of
    $\QD(\D_{xy}')$ that contains the positive $y$-axis. If
    we translate $Q_{top}$ to the highest point in $S_j \cap
    N$, then it does not contain a point of $N$ in its
    interior, implying that this point is a vertex of $C$,
    see Figure \ref{fig:squares}.
    
    This implies that the expected complexity of
    $\DCH{\D_{xy}'}(N)$ is $\Omega \pth{\sqrt{n}}$
\end{proof}

%%%%%%%%%%%%%%%%%%%%%%%%%%%%%%%%%%%%%%%%%%%%%%%%%%%%%%%%%%%%%%%%%%%%%%%%%
%%%%%%%%%%%%%%%%%%%%%%%%%%%%%%%%%%%%%%%%%%%%%%%%%%%%%%%%%%%%%%%%%%%%%%%%%

\section[Result]{On the Expected Number of Points on the Boundary of the
   Quadrant Hull Inside a Hypercube}
\label{sec:hcube}

In this section, we show that the expected number of points
on the boundary of the quadrant hull of a set $S$ of $n$
points, chosen uniformly and independently from the unit
cube is $O(\log^{d-1}n)$.  Those bounds are known
\cite{bkst-anmsv-78}, but we believe the new proof is
simpler.

\begin{definition}[\cite{mp-ofsch-97}]
    Let $\Q$ be a family of subsets of $\Re^d$. For a set $A
    \subseteq \Re^d$, we define the $\Q$-hull of $A$ as
    \[
    \QHull{\Q}(A) = \bigcap\brc{ Q \in \Q \sep{ A \subseteq
          Q }}.
    \]
\end{definition}

\begin{definition}[\cite{mp-ofsch-97}]
    For a sign vector $s \in \brc{-1, +1}^d$, define
    \[
    q_s = \brc{ x \in \Re^d \sep{ \sign(x_i ) = s_i,
          \text{ for } i=1, \ldots, d }},
    \]
    and for $a \in \Re^d$, let $q_s(a) =q_s + a$. We set
    $\Q_{sc} = \brc{ \Re^d \setminus q_s(a) \sep{ a \in
          \Re^d, s \in \brc{-1,+1}^d}}$. We shall refer to
    $\QH(A)$ as the {\em quadrant hull} of $A$. These are
    all points which cannot be separated from $A$ by any
    open orthant in space (i.e., quadrant in the plane).
\end{definition}

\begin{definition}
    Given a set of points $S \subseteq \Re^d$, a point $p
    \in \Re^d$ is {\em $\QSC$-exposed}, if there is $s \in
    \brc{-1,+1}^d$, such that $q_s(p) \cap S = \emptyset$. A
    set $C$ is {\em $\QSC$-exposed}, if there exists a
    point $p\in C$ which is $\QSC$-exposed.
\end{definition}

\begin{definition}
    For a set $S \subseteq \Re^d$, let $n_{sc}(S)$ denote
    the number of points of $S$ on the boundary of $\QH(S)$.
\end{definition} 

\begin{theorem}
    Let $\C$ be a unit axis parallel hypercube in $\Re^d$,
    and let $S$ be a set of $n$ points chosen uniformly and
    independently from $\C$. Then, the expected number of
    points of $S$ on the boundary of $H = \QH(S)$ is
    $O(\log^{d-1}(n))$.
    \label{theorem:main}
\end{theorem}

\begin{proof}
    We partition $\C$ into equal size tiles, of volume
    $1/n^d$; that is $C(i_1, i_2, \ldots, i_d) = [(i_1-1)/n,
    i_1/n] \times \cdots \times [(i_d - 1)/n, i_d/n]$, for
    $1 \leq i_1, i_2, \ldots, i_d \leq n$.
    
    We claim that the expect number of tiles in our
    partition of $\C$ which are exposed by $S$ is $O(n^{d-1}
    \log^{d-1}n)$.
    
    Indeed, let $q = q_{(-1, -1, \ldots, -1)}$ be the
    ``negative'' quadrant of $\Re^d$. Let $X(i_2, \ldots,
    i_d)$ be the maximal integer $k$, for which $C(k,i_2,
    \ldots, i_d)$ is exposed by $q$.  The probability that
    $X(i_2, \ldots, i_d) \geq k$ is bounded by the
    probability that the cubes $C(l_1, l_2, \ldots, l_d)$
    does not contain a point of $S$, where $l_1 < k, l_2<
    i_2, \ldots, l_d < i_d$. Thus,
    \begin{eqnarray*}
        \Pr \pbrc{ X(i_2, \ldots, i_d) \geq k} &\leq& \pth{
           1 - \frac{(k-1)(i_2
              -1) \cdots (i_d-1)}{n^d}}^n \\
        &\leq& \exp \pth{ -\frac{{(k-1)(i_2 - 1) \cdots
                 (i_d-1)}}{n^{d-1}}},
    \end{eqnarray*}
    since $1-x \leq e^{-x}$, for $x \geq 0$.
    
    Hence, the probability that $\Pr\pbrc{X(i_2, \ldots,
       i_d) \geq i\cdot m + 1} \leq e^{-i}$, where\\ $m =
    \ceil{\frac{n^{d-1}}{(i_2-1) \cdots (i_d-1)}}$. Thus,
    \begin{eqnarray*}
        E \pbrc{ X(i_2, \ldots, i_d) } &=&
        \sum_{i=1}^{\infty} i \Pr\pbrc{X(i_2, \ldots, i_d) =
           i } = \sum_{i=0}^{\infty} \sum_{j=im+1}^{(i+1)m}
        j \Pr\pbrc{X(i_2, \ldots, i_d) = j }\\
        &\leq& \sum_{i=0}^{\infty} (i+1)m \Pr \pbrc{ X(i_2,
           \ldots, i_d) \geq im + 1} \leq
        \sum_{i=0}^{\infty} (i+1)me^{-i} = O(m).
    \end{eqnarray*}
    
    Let $r$ denote the expected number of tiles exposed by
    $q$ in $\C$.  If $C(i_1, \ldots, i_d)$ is exposed by
    $q$, then $X(i_2, \ldots, i_d) \geq i_1$. Thus, one can
    bound $r$ by the number of tiles on the boundary of
    $\C$, plus the sum of the expectations of the variables
    $X(i_2, \ldots, i_d)$. We have
    \begin{eqnarray*}
        r &=& O(n^{d-1}) + \sum_{i_2=2}^{n-1}
        \sum_{i_3=2}^{n-1} \cdots \sum_{i_d=2}^{n-1} O \pth{
           \frac{n^{d-1}}{(i_2 - 1)(i_3 -1 )
              \cdots(i_d - 1)}} \\
        &=& O \pth{ n^{d-1}} \sum_{i_2=2}^{n-1}
        \frac{1}{i_2-1}\sum_{i_3=2}^{n-1} \frac{1}{i_3-1}
        \cdots \sum_{i_d=2}^{n-1} \frac{1}{i_d-1} = O \pth{
           n^{d-1} \log^{d-1}{n}}.
    \end{eqnarray*}
    
    The set $\Q_{sc}$ contains translation of $2^d$
    different quadrants.  This implies, by symmetry, that
    the expected number of tiles exposed in $\C$ by $S$ is
    $O \pth{ 2^dn^{d-1} \log^{d-1}{n}} = O \pth{n^{d-1}
       \log^{d-1}{n}}$. However, if a tile is not exposed by
    any $q_s$, for $s \in \brc{-1,+1}^d$, then it lies in
    the interior of $H$. Implying that the expected volume
    of $H$ is at least
    \[
    \frac{n^d - O\pth{n^{d-1} \log^{d-1}{n}}}{n^d} = 1 -
    O\pth{\frac{\log^{d-1} n}{n}}.
    \]
    
    We now apply an argument similar to the one used in
    Lemma \ref{lemma:area:to:vertices} (Efron's Theorem),
    and the theorem follows.
\end{proof}

\remove{    
    We are now in position to apply an argument similar to
    the one used in Lemma \ref{lemma:area:to:vertices}:
    Partition $S$ into two equal size sets $S_1, S_2$.  We
    have that $n_{sc}(S)$ is bounded by the number of points
    of $S_1$ falling outside $H_2 = \QH(S_2)$, plus the
    number of points of $S_2$ falling outside $H_1 =
    \QH(S_1)$, where $n_{sc}(S)$ is the number of points of
    $S$ on the boundary of $\QH(S)$. We conclude,
    \begin{eqnarray*}
        E[ n_{sc}(S) ] &\leq& E \pbrc{ |S_1 \setminus H_2 |
           } + E \pbrc{ |S_2 \setminus H_1|} = E\pbrc{
           \rule{0cm}{0.5cm} E \pbrc{ |S_1 \setminus H_2|
              \sep{ S_2 } }} + E \pbrc{ \rule{0cm}{0.5cm} E
           \pbrc{ |S_2
              \setminus H_1| \sep{ S_1 } } }\\
        &=& E\pbrc{ \frac{n}{2}\pth{1 - \Vol(H_2)} +
           \frac{n}{2}\pth{1 - \Vol(H_1)}} = O \pth{ n \cdot
           \frac{\log^{d-1} n}{n}} = O( \log^{d-1} n),
    \end{eqnarray*}
    and the theorem follows.
}

\begin{remark}
    A point $p$ of $S$ is a {\em maxima}, if there is no
    point $p'$ in $S$, such that $p_i \leq p'_i$, for
    $i=1,\ldots, d$. Clearly, a point which is a maxima, is
    also on the boundary of $\QH(S)$. By Theorem
    \ref{theorem:main}, the expected number of maxima in a
    set of $n$ points chosen independently and uniformly
    from the unit hypercube in $\Re^d$ is $O(\log^{d-1} n)$.
    This was also proved in \cite{bkst-anmsv-78}, but we
    believe that our new proof is simpler.
    
    Also, as noted in \cite{bkst-anmsv-78}, a vertex of the
    convex hull of $S$ is a point of $S$ lying on the
    boundary of the $\QH(S)$. Hence, the expected
    number of vertices of the convex hull of a set of $n$
    points chosen uniformly and independently from a
    hypercube in $\Re^d$ is $O(\log^{d-1} n)$.
\end{remark}

%%%%%%%%%%%%%%%%%%%%%%%%%%%%%%%%%%%%%%%%%%%%%%%%%%%%%%%%%%%%%%%%%%%%%%%%
%%%%%%%%%%%%%%%%%%%%%%%%%%%%%%%%%%%%%%%%%%%%%%%%%%%%%%%%%%%%%%%%%%%%%%%%

\subsection*{Acknowledgments}

I wish to thank my thesis advisor, Micha Sharir, for his
help in preparing this manuscript. I also wish to thank
Pankaj Agarwal, and Imre B{\'a}r{\'a}ny for helpful
discussions concerning this and related problems.

%-------------------------------------------------------------------------
% Bibiliography 
%-------------------------------------------------------------------------

%\bibliographystyle{scalpha}
\bibliographystyle{salpha}
\bibliography{geometry}

\end{document}